\documentclass[11pt,onecolumn]{article}
\usepackage{times}
\usepackage{amsmath}
\usepackage{amssymb}
\usepackage{xspace}
\usepackage{theorem}
\usepackage{graphicx}
\usepackage{epsfig}
\usepackage{ifpdf}
\usepackage{url,hyperref}
\usepackage{latexsym}
\usepackage{euscript}
\usepackage{xspace}
\usepackage{color}
\usepackage{makeidx}
\usepackage{picins,wrapfig}
\usepackage{stackrel}
\usepackage{amscd}
\usepackage[all]{xy}

\long\def\remove#1{}

\newtheorem{theorem}{Theorem}[section] 

\newtheorem{claim}[theorem]{Claim}

\newtheorem{cor}[theorem]{Corollary}
\newtheorem{fact}[theorem]{Fact}
\newtheorem{observation}[theorem]{Observation}
\newtheorem{proposition}[theorem]{Proposition}
\newtheorem{definition}[theorem]{Definition}
\newenvironment{proof}{{\em Proof:}}{\hfill{\hfill\rule{2mm}{2mm}}}

\newcommand {\mm}[1] {\ifmmode{#1}\else{\mbox{\(#1\)}}\fi}
\newcommand{\denselist}{\itemsep 0pt\parsep=1pt\partopsep 0pt}
\newcommand{\eps}{{\varepsilon}}
\newcommand{\reals}	{{\rm I\!\hspace{-0.025em} R}}

\newcommand{\resdel}[2]{{\rm Del}|_{#2}\, #1}
\newcommand{\vor}[1]{{\rm Vor\,}#1}
\newcommand{\del}[1]{{\rm Del\,}#1}

\newcommand{\B}		{{\cal B}}
\newcommand{\K}		{{\cal K}}

\newcommand{\G}		{\mathcal{G}}

\newcommand{\Rips}	{{\cal R}}
\newcommand{\homo}	{{\sf H}}

\DeclareMathOperator{\argmin} {\mathrm argmin}

\begin{document}

\title{Graph Induced Complex on Point Data}

\author{
Tamal K. Dey\thanks{
Department of Computer Science and Engineering,
The Ohio State University, Columbus, OH 43210, USA.
Email: {\tt tamaldey@cse.ohio-state.edu}}
\quad\quad
Fengtao Fan\thanks{
Department of Computer Science and Engineering,
The Ohio State University, Columbus, OH 43210, USA.
Email: {\tt fanf@cse.ohio-state.edu}}
\quad\quad Yusu Wang\thanks{
Department of Computer Science and Engineering,
The Ohio State University, Columbus, OH 43210, USA.
Email: {\tt yusu@cse.ohio-state.edu}}
}

\date{}
\maketitle

\begin{abstract}
The efficiency of extracting topological information from point
data depends largely on the complex that is built on top of the data
points. From a computational viewpoint, the most favored  
complexes for this purpose have so far been
Vietoris-Rips and witness complexes. While the Vietoris-Rips complex
is simple to compute and is a good vehicle 
for extracting topology of sampled spaces, 
its size is huge--particularly in high dimensions.
The witness complex on the other hand enjoys a smaller size because 
of a subsampling, but fails to capture the 
topology in high dimensions unless imposed with extra structures. 
We investigate a complex called the {\em graph induced complex}
that, to some extent, enjoys the advantages of both. It works on a subsample
but still retains the power of capturing the topology as the
Vietoris-Rips complex. It only needs a graph connecting
the original sample points from which it builds a 
complex on the subsample thus taming the
size considerably. We show that, using the graph induced complex one
can (i) infer the one dimensional homology of a manifold from a
very lean subsample, (ii) reconstruct a surface in three dimension
from a sparse subsample 
without computing Delaunay triangulations, (iii) infer the 
persistent homology groups of compact sets from a sufficiently
dense sample. We provide experimental evidences in support of
our theory. 
\end{abstract}

\newpage
\setcounter{page}{1}
\section{Introduction}
Acquiring knowledge about a sampled space from a point data has become
a key problem in many areas of science and engineering. The sampled space
could be a hidden manifold sitting in some high dimensions, or could be
a compact subset of some Euclidean space. 
Topological information such as the rank of the homology groups,
or their persistent behavior can divulge important features of
the hidden space. Therefore, a considerable effort has
ensued to extract topological information
from point data in recent years~\cite{CGOS09,CO08,DW11,Ghrist}. 
With the advent of advanced
technologies, the data is often generated in abundance. Mixed with the
burden of high dimensionality, large data sets pose a challenge to the 
processing resource. As a result, some recent investigations have
focused on how to use a lighter data structure or sparsify the input
which aids a faster computation, but still guarantees that the output
inference is correct.

Point data by themselves do not have interesting topology. So, a foremost
step in topology inference is to impose a structure such as a simplicial
complex onto it. The Delaunay, \v{C}ech, Vietoris-Rips, and witness complexes
are some of the most commonly proposed complexes 
for this purpose.
Among these, Vietoris-Rips (Rips in short) and witness complexes~\cite{CO08}
have been favored because they can be constructed
with simple computations.
Rips complexes are easy to construct as they can be built from
a graph by recognizing the cliques in it. However, the presence of
simplices corresponding to all cliques makes its size quite large.
Even in three dimensions with a few thousand points, the size of the
Rips complex can be an obstacle, if not a stopper, for further processing. 
Witness complexes, on the other hand, have too few simplices 
to capture the topology of the sampled space in dimensions three
or more~\cite{BGO09}. To tackle this issue, Boissonnat et al.~\cite{BGO09} 
suggested modifications to the original definition
of witness complex~\cite{SC04}. This
enlarges the witness complexes but 
makes it more complicated and costly to compute. 

\begin{figure}[htb]
\begin{center}
\includegraphics[height = 4.5cm]{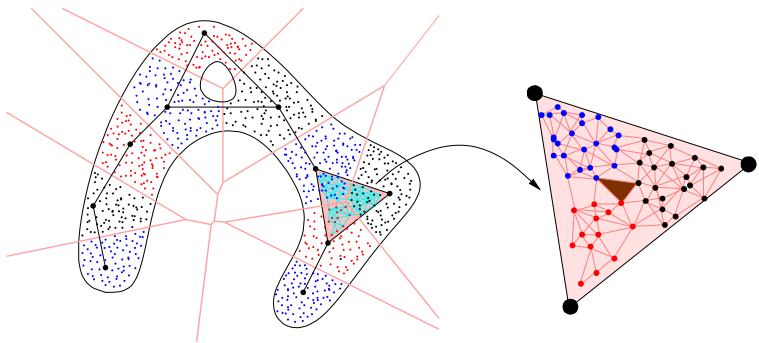}
\caption{A graph induced complex shown with bold vertices, edges,
 and a shaded triangle
on left. Input graph within the shaded triangle is shown on right. 
}
\label{fig:gic}
\end{center}
\end{figure}
\vspace*{-0.08in}
We investigate a new complex, a version of which was originally introduced in~\cite{Fang} for the application of sensor network routing. 
We set up a more general definition and call it the 
{\em graph induced complex}. 
We provide new theoretical understanding of the graph induced complex 
in terms of topology inference. 
In particular, we show that, when equipped with 
appropriate metric, this complex 
can decipher the topology from data. 
It retains the simplicity of the Rips complex as well as
the sparsity of the witness
complex. Its construction resembles
the sparsified Rips complex
proposed in~\cite{Sheehy} and also the combinatorial
Delaunay triangulation proposed in~\cite{CS03}, but it does
not build a Rips complex on the subsample and thus is sparser
than the Rips complex with the same set of vertices.
This fact makes a real difference in practice
as our preliminary experiments show. 
The idea of graph induced complex also bears similarity to the geodesic Delaunay triangulation which was proposed to recover the topology of a bounded planar region (with holes) from point samples \cite{GGOW10}. Our work extends it, as well as investigates its theoretical properties, to more general setting beyond the planar case. 

Given a graph $G$ on a point data $P$ equipped with a metric, 
one can build a graph induced
complex on a subsample $Q\subseteq P$ by throwing in
a simplex with a vertex set $V \subseteq Q$ if a set of points
in $P$, each being closest to exactly one vertex in $V$, forms a clique.
Figure~\ref{fig:gic} shows a graph induced complex for a point data
in the plane. Subsampled points are the darker vertices. Input points
are grouped according to the proximity to the subsampled vertices
(indicated with a Voronoi partition). The shaded triangle enlarged on right
is in the graph induced complex
since there is a $3$-clique in the input graph whose 3 vertices 
have 3 different closest point in the subsample.
Observe that, in this example, 
the graph induced complex has the same homology as the
sampled space. 

Figure~\ref{fig:comparison} shows
experimental results on two data sets, 40,000 sample points
from a Klein bottle in $\mathbb{R}^4$ and 15,000 sample points 
from the primary circle of natural image data 
considered in $\mathbb{R}^{25}$~\cite{AC09}.
The graphs connecting any two points within $\alpha=0.05$ unit distance
for Klein bottle and $\alpha=0.6$ unit distance for the primary circle
were taken as input for the graph induced complexes.
The $2$-skeleton of the Rips complexes for these $\alpha$ parameters
have $608,200$ and $1,329,672,867$ simplices respectively. These sizes
are too large to carry out fast computations.

\begin{figure*}[h!]
\begin{center}
\begin{tabular}{cc}
        \includegraphics[width=0.4\textwidth]{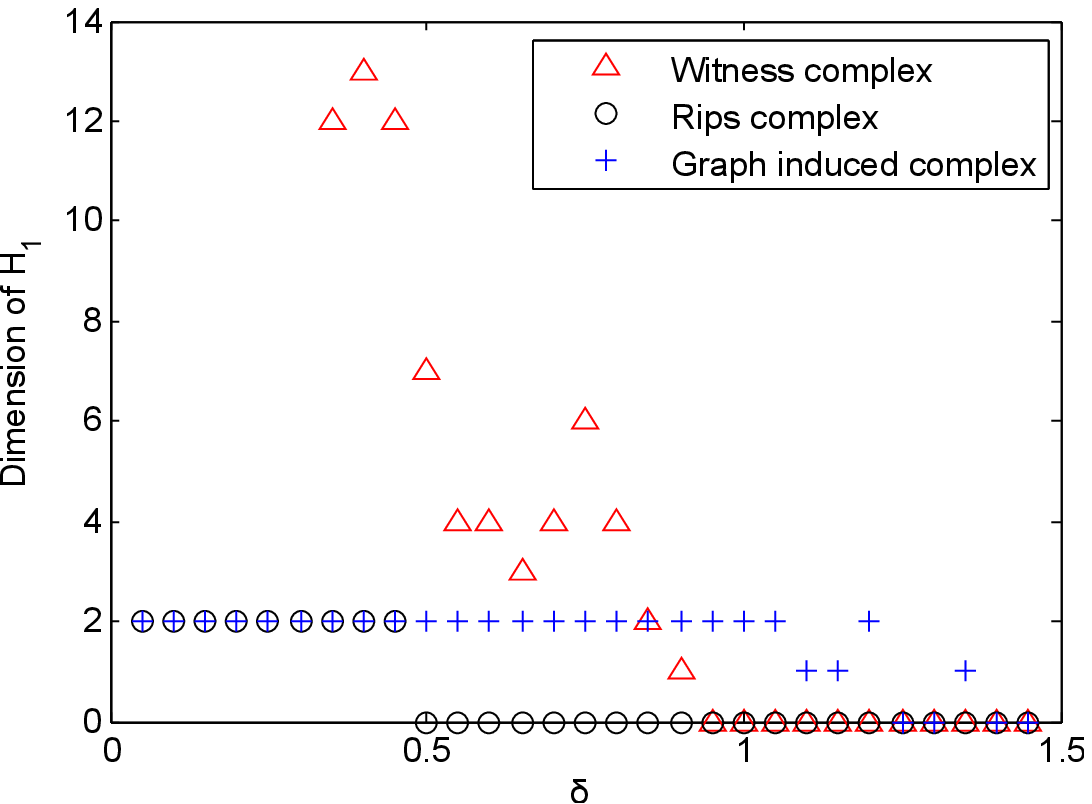}&
         \includegraphics[width=0.4\textwidth]{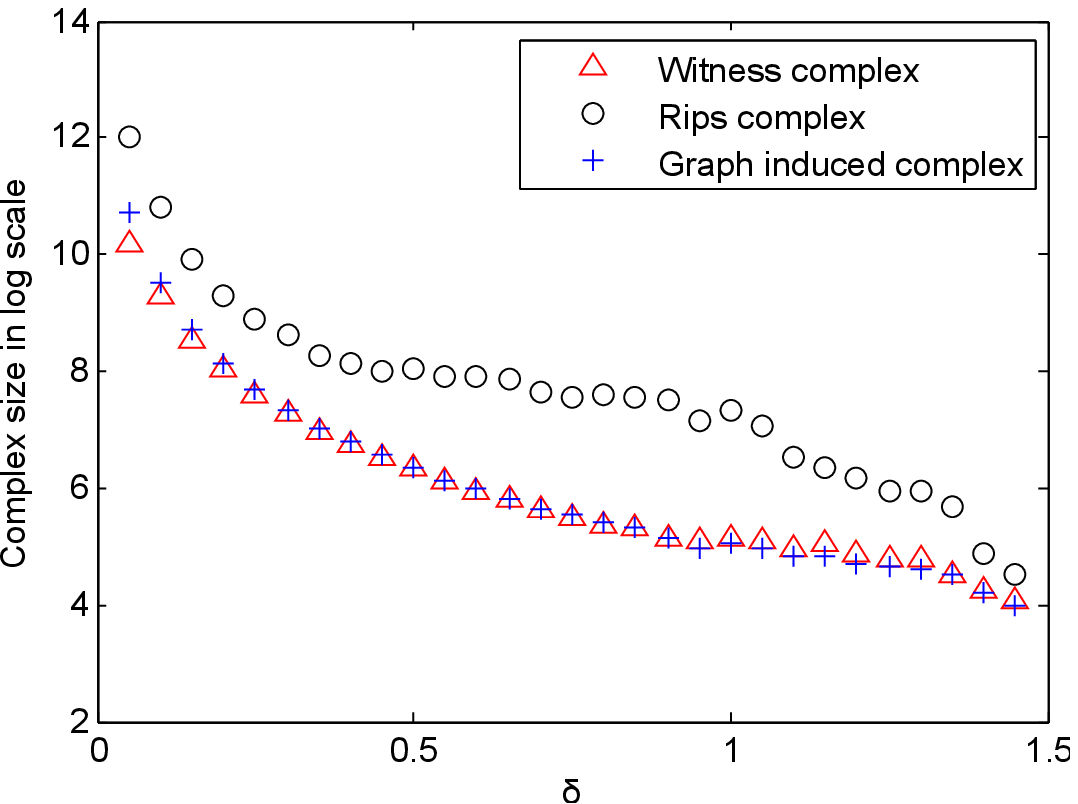}\\
        (a) & (b)\\
        \includegraphics[width=0.4\textwidth]{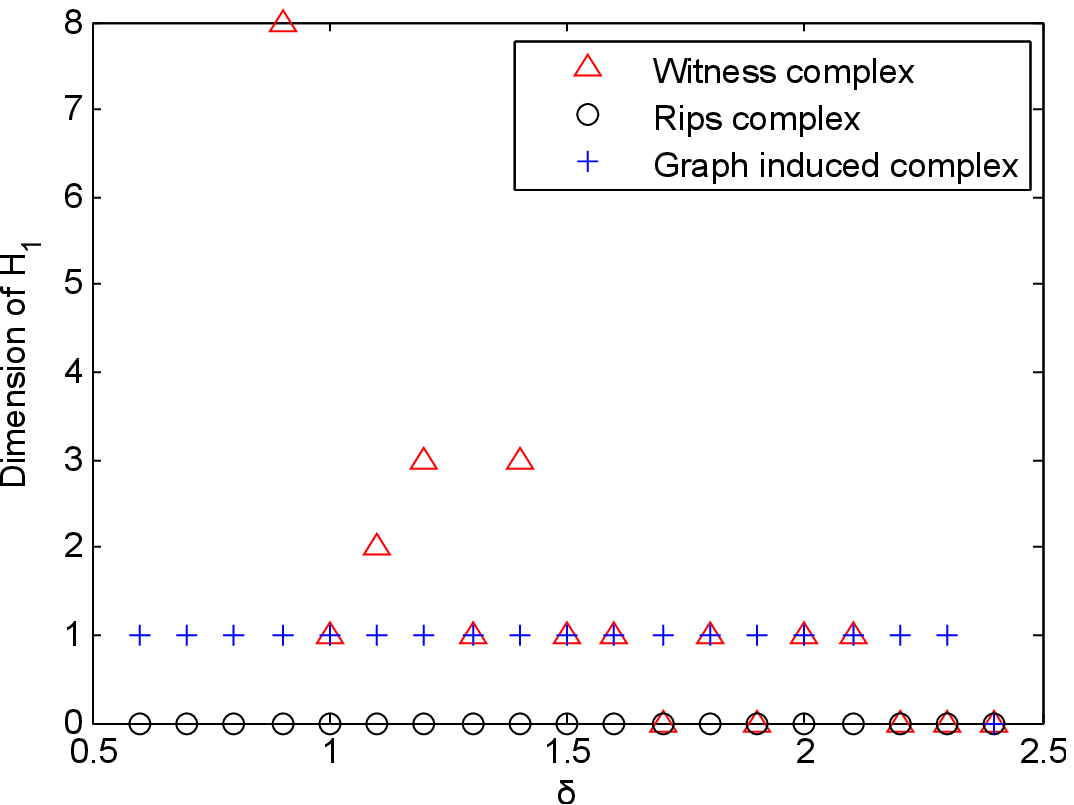}&
        \includegraphics[width=0.4\textwidth]{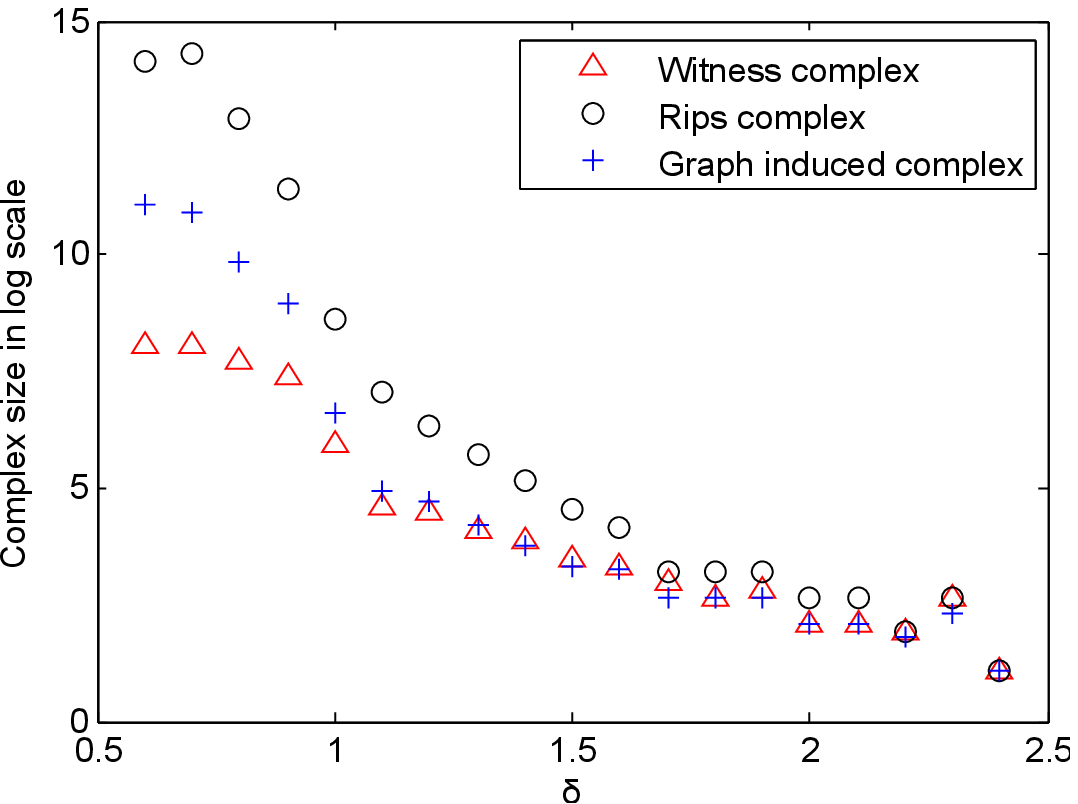}\\
        (c) & (d)
\end{tabular}
\end{center}
\caption{Comparison results for Klein bottle 
in $\mathbb{R}^{4}$ (top row) and primary circle in 
$\mathbb{R}^{25}$ (bottom row).
The estimated $\beta_1$ for three complexes are shown
on the left, and their sizes are shown on log scale on right. 
}
\label{fig:comparison}
\end{figure*}

For comparisons, we constructed the graph induced complex, sparsified Rips complex, and
the witness complex on the same subsample determined by a parameter
$\delta$. The parameter $\delta$ is also used in the graph 
induced complex (see definitions later) and the witness complex. The edges
in the Rips complex 
built on the same \emph{subsample} were of 
lengths at most $\alpha+2\delta$ (justified by Proposition~\ref{prop:contg}). 
We varied $\delta$ and observed the
rank of the one dimensional homology group ($\beta_1$).
As evident from the plots, the graph induced complex
captured $\beta_1$ correctly for a significantly
wider range of $\delta$ (left plots) while its size
remained comparable to that of the witness complex (right plots).
In some cases, the graph induced complex could capture the
correct $\beta_1$ with remarkably small number of simplices. 
For example, 
it had $\beta_1=2$ for Klein bottle when there
were $278$ simplices for $\delta=0.7$
and $154$ simplices for $\delta=1.0$. In both cases Rips
and witness complexes had wrong $\beta_1$ while the Rips complex
had a much larger size ($\log_e$ scale plot) and the witness complex
had comparable size.
This illustrates why the graph induced complex can be 
a better choice than the Rips and witness complexes. 

We establish three different results. First, we show that
the one-dimensional homology group of surfaces in three dimensions can be
determined by graph induced complexes. Even the surface itself can be reconstructed
with some post-processing from a sparse subsample
of a sample that could be excessively dense. Second, we show that,
for higher-dimensional manifolds, one-dimensional homology can still
be determined from graph induced complexes with a simple modification
of the metric. Finally, we extend our results to other homology groups where 
we show that the persistent homology groups of a pair of graph
induced complexes can determine the homology groups
of compact spaces. Experimental results support our theory.

\section{Graph induced complex and preliminaries}
First we define the graph induced complex in a more abstract
setting which does not require any metric.
\begin{definition} Let $G(V)$ be a graph with the
vertex set $V$ and let $\nu: V\rightarrow V'$ be a vertex
map where $\nu(V)=V'\subseteq V$.
The graph induced complex $\G(V,V',\nu)$ is defined
as the simplicial complex where a $k$-simplex 
$\sigma=\{v_1',v_2',\ldots,v_{k+1}'\}$
is in $\G(V,V',\nu)$ if and only if there exists a $(k+1)$-clique 
$\{v_1,v_2,\ldots,v_{k+1}\}\subseteq V$ so that $\nu(v_i)=v_i'$ for
each $i\in \{1,2,\ldots,k+1\}$.
To see that it is indeed a simplicial complex,
observe that a subset of a clique is also a clique.
\end{definition}

Now we specialize the graph induced complex to the case where
the vertices of the input graph comes from a metric space.

\begin{definition} A metric space $(X,d)$ is a tuple where
$X$ is a set and 
$d: X\times X \rightarrow \mathbb{R}_+$
is a distance function satisfying $d(x,y)\geq 0$,
$d(x,y)=0$ iff $x=y$, 
$d(x,y)=d(y,x)$, and $d(x,y)\leq d(x,z)+d(z,y)$.
\end{definition}

\begin{definition}
Let $(P,d)$ be a metric space where
$P$ is a finite point set and let $Q \subset P$ be a subset.
Let $\nu_d: P\rightarrow Q$ denote the nearest point map
where $\nu_d(p)$ is a point in $\argmin_{q\in Q} d(p,Q)$.
Given a graph $G(P)$ with $P$ as its vertex set, we
define its graph induced complex as
$\G(P,Q,d):=\G(P,Q,\nu_d)$.
\end{definition}


Among the many possible choices for $d$, we will  
focus on two cases where $d=d_E$, the Euclidean distance, and $d=d_G$, the
graph distance induced by the graph $G(P)$ assuming its edges have
non-negative weights. 
For any two vertices $p_1,p_2\in P$, the distance $d_G(p_1,p_2)$ is the
length of the shortest path between $p_1$ and $p_2$ in $G(P)$.
We will describe the choices of the distance functions 
as and when necessary.

In our case, the point set $P$ will be a discrete subset of
a compact smooth manifold $M\subset \mathbb{R}^n$ without boundary,
or simply of a compact set $X\subset \mathbb{R}^n$. 
The graph $G(P)$ will be the graph
$G^{\alpha}(P)=(P, E^{\alpha})$ 
where $(p_1,p_2)\in E^{\alpha}$ if and only if $\|p_1-p_2\|\leq \alpha$.
The graph induced complex induced by $G^{\alpha}(P)$ 
on a subset $Q\subseteq P$ under a distance function $d$
will be the focus of our
study. To emphasize the dependence on the parameter $\alpha$, 
we denote it as $\G^{\alpha}(P,Q,d)$.
One may draw a parallel between the graph induced complexes and
the well-known witness complexes~\cite{SC04} where
$P$ acts as a witness set $W$ and $Q$ acts as a landmark set $L\subseteq W$. 
However, the analogy does not extend any further since the construction
of the witness complex and its variants~\cite{BGO09} differs 
from that of the graph induced complex. For example, the original
witness complex defined in~\cite{SC04} embraces a $k$-simplex with vertex set
in $L$ only if its vertices belong to the $k$-nearest neighbors
of a point in $W$. In contrast the graph induced complex 
embraces a $k$-simplex only if its vertices have nearest
neighbors in $W$ that form a $k$-clique in
a graph built on the vertices belonging to $W$.
Similar to the witness
complexes, the graph induced complex builds upon a subsampling.
But, unlike witness complexes,
it enjoys some topological properties without any extra 
modifications such as weighting \cite{BGO09}.

\subsection{Sampling, homology, and sandwiching} 
As indicated before, the input point set
$P$ is a sample of a smooth manifold  $M$ or of a compact set $X$ embedded in
an Euclidean space. We will also subsample $P$ according to
a distance function $d$. Therefore, we define sampling in
a more general context. 

\begin{definition} A finite set $P\subseteq X$ 
is an $\eps$-sample of a metric space $(X,d)$, 
if for each point $x\in X$ there is
a point in $p\in P$ so that $d(x,p)\leq \eps$.
Additionally, $P$ is called $\delta$-sparse if $d(p_1,p_2)\geq \delta$
for any pair of points in $P$.  
\end{definition}

\noindent The point set $P$ does not have interesting topology by itself. 
We build
simplicial complexes using $P$ as the vertex set to infer the topology
of the sampled space $X$. Specifically, our goal is to infer
the homology groups of a manifold or a compact set from which $P$ is sampled
by computing the homology groups of a simplicial complex built with $P$
as vertices.  Let $\homo_r(\cdot)$ denote the $r$-dimensional homology group. It
refers to the singular homology when the argument is a manifold or 
a compact set, 
and to the simplicial homology when it is a simplicial complex. Also, all
homology groups are assumed to be defined over the finite field $\mathbb{Z}_2$. 

Our main tool for topological inference rests 
on the relations of the graph induced complexes to the Rips complexes 
that are known to capture
information about the homology groups of spaces~\cite{ALS11,Haus95}.
\begin{definition}
Given a point set $P\subseteq \mathbb{R}^n$ and a parameter $\alpha$, 
the Rips complex $\Rips^{\alpha}(P) =\Rips^{\alpha}(P,d_E)$ is a 
simplicial complex 
where a simplex $\sigma\in \Rips^{\alpha}(P)$ if and only if
all vertices of $\sigma$, drawn from  $P$, are within $\alpha$
Euclidean distance of each other. 
\end{definition}
Notice that we define Rips complexes with Euclidean distances
instead of general metrics
which will be assumed throughout this paper.
It is known that such Rips complexes capture the
topology of a manifold $M$ if the parameters are chosen right~\cite{ALS11,Haus95}. We utilize
this fact to infer $\homo_1(M)$ by exploiting a sandwiching property of 
graph induced complexes by Rips complexes. 
To prove this fact, we recall 
the concept of contiguous maps
from algebraic topology. Our main interest in this concept is the fact that
two contiguous maps between two simplicial
complexes induce the same homomorphism at the homology level.
\begin{definition}[\cite{Munkres}]
Let $\K_1$ and $\K_2$ be two simplicial complexes connected by two
simplicial maps $a: \K_1\rightarrow \K_2$ and $b:\K_1\rightarrow \K_2$.
We say $a$ and $b$ are contiguous, if and only if for any 
simplex $\sigma\in \K_1$, the simplices $a(\sigma)$ and $b(\sigma)$ are
faces of a common simplex in $\K_2$.
\end{definition}

\begin{fact}[\cite{Munkres}]
If $a:\K_1 \rightarrow \K_2$ and $b:\K_1\rightarrow \K_2$ are
contiguous, then the induced homomorphisms $a_*: \homo_r(\K_1)\rightarrow \homo_r(\K_2)$ 
and $b_*: \homo_r(\K_1)\rightarrow \homo_r(\K_2)$ are equal.
\end{fact} 

In our case two simplicial complexes will be $\K_1=\Rips^{\alpha}(P)$
and $\K_2=\Rips^{\beta}(P)$ for some $\beta > \alpha$. The map $a$ is
an inclusion $\Rips^{\alpha}(P)\hookrightarrow \Rips^{\beta}(P)$. 
For the map $b$, we consider  
a simplicial map 
$h: \Rips^{\alpha}(P)\rightarrow \G^{\alpha}(P,Q,d)$ which composed
with an inclusion $\G^{\alpha}(P,Q,d)\hookrightarrow \Rips^{\beta}(P)$ 
provides $b$. We elaborate on this construction.

The vertex sets of $\Rips^{\alpha}(P)$ and
$\G^{\alpha}(P,Q,d)$ are $P$ and $Q$ respectively with a vertex map
$\nu:P\rightarrow Q$ where $p\in P$ maps to one
of its closest point $\nu(p)\in Q$
with respect to the distance function $d$.
Observe that $G^{\alpha}(P)$ is the $1$-skeleton of
$\Rips^{\alpha}(P)$. Therefore, the edges of a $(k+1)$-clique in $G^{\alpha}(P)$
constitute the $1$-skeleton of a $k$-simplex 
in $\Rips^{\alpha}(P)$ and vice versa. 
The vertex map $\nu$ extends to a simplicial map 
$h:\Rips^{\alpha}(P)\rightarrow \G^{\alpha}(P,Q,d)$ where
a $k$-simplex $\{p_1,p_2,\cdots,p_{k+1}\}$ in $\Rips^{\alpha}(P)$
is mapped to a simplex (of dimension at most $k$) with the vertex set $\bigcup_i\{\nu(p_i)\}$.
To see that $h$ is well defined, observe that any subset
of the $(k+1)$-clique $\{p_1,p_2,\cdots,p_{k+1}\}$ is also a 
clique in $G^{\alpha}(P)$ and hence $\bigcup_i\{\nu(p_i)\}$ is
a simplex in $\G^{\alpha}(P,Q,d)$.
The following result is used later.

\begin{proposition} \label{prop:contg}
Let $(P,d)$ be a metric space where $P\subset \mathbb{R}^n$ 
is a finite set and for every pair $p_1,p_2\in P$, 
$d(p_1,p_2)$ is at least the
Euclidean distance $\|p_1-p_2\|$. Let $Q$ be a $\delta$-sample
of $(P,d)$. 
We have the sequence
\[
\Rips^\alpha(P) \stackrel{h}{\longrightarrow}  
\G^\alpha(P,Q,d)  \stackrel{j}{\hookrightarrow}  \Rips^{\alpha+2\delta}(P)
\]
where $j$ is an inclusion and $j\circ h$ is contiguous to the inclusion
$i: \Rips^{\alpha}(P) \hookrightarrow \Rips^{\alpha+2\delta}(P)$.
Hence, $j_* \circ h_* = i_*$.
\end{proposition}
\begin{proof}
The map $h$ is well-defined as we detailed before. We observe that
$\G^{\alpha}(P,Q,d) \subseteq \Rips^{\alpha+2\delta}(P)$ because
any edge $(q_1,q_2)$ of a simplex $\sigma\in \G^{\alpha}(P,Q,d)$ 
satisfies $d(q_1,q_2)\leq \alpha+2\delta$. Since $d(q_1,q_2)\geq \|q_1-q_2\|$
by assumption, the edge $(q_1,q_2)$ and hence the simplex $\sigma$ 
are in $\Rips^{\alpha+2\delta}(P)$. 
It follows that the inclusion map $j$ is well-defined.

To prove the contiguity, consider a simplex
$\sigma$ in $\Rips^{\alpha}(P)$. 
We need to show that the vertices of $\sigma$ and $h(\sigma)$ 
span a simplex in $\Rips^{\alpha+2\delta}(P)$.
Clearly, all vertices of $\sigma$ are within $\alpha$ distance of each other.
By definition of $h$,
all vertices of $h(\sigma)$ are within distance $\alpha+2\delta$.
Let $u$ be a vertex of $\sigma$ and $h(v)$ be a vertex of $h(\sigma)$ where
$v$ is a vertex of $\sigma$.
Then the Euclidean distance  $\|u-h(v)\|$ is
at most $\|u-v\| + \|v-h(v)\| \leq \alpha + \delta$.
Therefore, all vertices of $\sigma$ and
$h(\sigma)$ are within $\alpha+2\delta$ distance.
Hence, the simplex $\sigma$ and $h(\sigma)$ are faces
of a common simplex in $\Rips^{\alpha+2\delta}(Q)$ proving the claim
of contiguity.
\end{proof} 

One may wonder how to efficiently construct 
the graph induced complexes in practice. Our experiments show that the
following procedure runs quite efficiently in practice.
It takes advantage of computing 
nearest neighbors within a range and, more importantly,
computing cliques only in a sparsified graph.

Let the ball $B(q,\delta)$ in metric $d$
be called the $\delta$-cover for the point $q$.
A graph induced complex $\G^{\alpha}(P,Q,d)$ where $Q$ is
a $\delta$-sparse $\delta$-sample can be built easily by
identifying $\delta$-covers with a rather standard iterative algorithm
similar to the greedy (farthest point) iterative algorithm of~\cite{Gon85}. 
Let $Q_i=\{q_1,\ldots,q_i\}$ be the point set sampled so far 
from $P$. We maintain the invariants (i) $Q_i$ is $\delta$-sparse
and (ii) every point $p\in P$ that are in the union of $\delta$-covers
$\bigcup_{q\in Q_i}B(q,\delta)$ have their closest point 
$\nu(p)=\argmin_{q\in Q_i} d(p,q)$
in $Q_i$ identified.
To augment $Q_i$ to $Q_{i+1}=Q_i\cup\{q_{i+1}\}$, we choose
a point $q_{i+1}\in P$ that is outside 
the $\delta$-covers $\bigcup_{q\in Q_i} B(q,\delta)$.
Certainly, $q_{i+1}$ is at least $\delta$ units away from all
points in $Q_i$ thus satisfying the first invariant. For the
second invariant, we check every point $p$ in the $\delta$-cover of
$q_{i+1}$ and update $\nu(p)$ 
to be $q_{i+1}$ if its distance to $q_{i+1}$ is smaller than
the distance $d(p,\nu(p))$. At the end, we obtain a sample
$Q\subseteq P$ whose $\delta$-covers cover the entire point set $P$
and thus is a $\delta$-sample of $(P,d)$ which is also $\delta$-sparse.
Next, we construct the simplices of $\G^{\alpha}(P,Q,d)$.
This needs identifying cliques in $G^{\alpha}(P)$
that have vertices with different closest points in $Q$.
We delete every edge $pp'$ from $G^{\alpha}(P)$ where
$\nu(p)=\nu(p')$. Then, we determine every clique $\{p_1,\ldots p_k\}$
in the remaining sparsified graph and include the simplex
$\{\nu(p_1),\ldots,\nu(p_k)\}$ in $\G^{\alpha}(P,Q,d)$. 
The main saving here is that many cliques of the original
graph are removed before it is processed for clique computation.
We use the recently proposed
simplex tree which computes cliques efficiently both time
and space-wise~\cite{BM12}.

\section{Surface point data}
In this section, we infer the geometry and topology of a surface
through the graph induced complex. 
Let $M$ be a smooth, compact surface embedded in $\mathbb{R}^3$
that has no boundary. We assume that $M$ has positive reach $\rho=\rho(M)$
which is the minimum distance of $M$ 
to its medial axis.
Let $P$ be an $\eps$-sample of the metric
space $(M,d_E)$ where $d_E$ is the Euclidean distance.
Consider the graph induced complex $\G^{\alpha}(P,Q,d_E)$.
In this section, the subset $Q\subset P$ is assumed to
be a $\delta$-sparse $\delta$-sample of $(P,d_E)$.

Our result in this section is that under certain conditions on $\alpha$,
$\eps$ and $\delta$,
$\G^{\alpha}(P,Q,d_E)$ captures the 
homology of $M$ and contains the restricted Delaunay 
triangulation $\resdel{Q}{M}$ as defined below. 
The sparsity of $Q$ turns out to be a crucial condition in the argument.

\begin{definition}
Let $\del Q$ denote the Delaunay triangulation of a point set 
$Q\subset \mathbb{R}^3$.
The restricted Delaunay triangulation of $Q$ with respect
to a manifold $M\subset \mathbb{R}^3$, denoted $\resdel{Q}{M}$, is
defined to be the subcomplex of
$\del Q$ formed by all Delaunay simplices
whose Voronoi duals intersect $M$.
\end{definition}
 
\subsection{Topological inference from $\G^{\alpha}(P,Q,d_E)$}
\label{surface-topoinfer}
Consider the sequence
$
\Rips^\alpha(P) \stackrel{h}{\longrightarrow}  
\G^\alpha(P,Q,d)  \stackrel{j}{\hookrightarrow}  \Rips^{\alpha+2\delta}(P)
$
in Proposition~\ref{prop:contg}.
When $P$, an $\eps$-sample of $(M,d_E)$, is sufficiently dense, it is known
that $i_*: \homo_1(\Rips^{\alpha}(P))\rightarrow \homo_1(\Rips^{\beta}(P))$ 
is an isomorphism for appropriate $\alpha$ and $\beta$.
The homomorphism $h_*$ becomes injective 
if $i_*$ is an isomorphism since $i_*=j_*\circ h_*$. 
If we can show that $h_*$ is also surjective, then
$h_*$ becomes an isomorphism. We now show that $h_*$ is
indeed surjective for $\homo_1$-homology and hence
information about $\homo_1(M)$ can be obtained by 
computing $\homo_1(\G^{\alpha}(P,Q,d_E))$. 
First, we observe
the following. Let $P_q\subseteq P$ be the points who have $q\in Q$ as the
closest point. Notice that by the definitino of $h$, $h(P_q) = \{q\}$.
To prove that $h_*$ is surjective, it is sufficient
to prove that the preimage of
each 1-cycle in $\G^\alpha(P,Q,d_E)$ contains a
1-cycle of $\Rips^{\alpha}(P)$. This, in turn, is true if
the 1-skeleton of $\Rips^\alpha(P_q)$ is connected.

\begin{proposition}
If the $1$-skeleton of $\Rips^{\alpha}(P_q)$ is connected for all
$q\in Q$, then $h_*$ is surjective.
\label{surjectivity}
\end{proposition}
\begin{proof}
We show that the chain map $h_\#$ induced by the simplicial
map $h$ is surjective. It follows that the homomorphism
$h_*$ induced at the homology level is also surjective.
Let $c=q_0q_1+ q_1q_2+\cdots+q_kq_0$ be any 
$1$-cycle in $\G^\alpha(P,Q,d_E)$. The edges
$q_{i-1}q_i$ and
$q_iq_{i+1}$ have edges, say $p_{i-1}p_i'$ and $p_ip_{i+1}'$ 
respectively, in their preimage under $h$ in $\Rips^{\alpha}(P)$.
Consider a path $\gamma_i$ between
$p_{i}$ and $p_i'$ in $\Rips^{\alpha}(P_{q_i})$ where
$h(p_{i})=h(p_{i}')=q_i$. Such a path exists because $\Rips^{\alpha}(P_q)$
is connected for all $q\in Q$. We have a $1$-cycle
$$
c'=p_0p_1' + \gamma_1 +p_1p_2'+\gamma_2+p_2p_3'+\cdots + 
\gamma_{k}+ p_{k}p_0'+\gamma_0
$$ 
in $\Rips^{\alpha}(P)$ so that $h_\#(c')=c$. 
This shows that $h_{*}$ is surjective in the first homology group.
\end{proof}

The $1$-skeleton of $\Rips^{\alpha}(P_q)$ is connected if the union
of balls $\B_q=\bigcup_{P_q}B(p,\frac{\alpha}{2})$ is connected because
an edge $p_1p_2$ is in $\Rips^{\alpha}(P_q)$ if the respective
balls $B(p_1,\frac{\alpha}{2})$ and $B(p_2,\frac{\alpha}{2})$
intersect. 
Let $V_q$ be the Voronoi cell of
$q$ in the Voronoi diagram $\vor{Q}$. Let 
$M_q = V_q \cap M$ be the restricted Voronoi region. 
It turns out (we will prove it later in Proposition~\ref{Rips-connect}) that if $M_q$ is contained in $\B_q$ and $M_q$ is connected, then $\B_q$ is connected. 
It may seem a priori that $\B_q$ would contain $M_q$ if 
$P$ is a dense sample. Unfortunately, that is not
true as Figure~\ref{fig:disconnect_skeleton} illustrates. To avoid
such a case, we require that the Voronoi cells
do not subtend very small angles between their facets
which is ensured by the $\delta$-sparsity of $Q$. 
Proposition~\ref{LEMMA:FIND-POINT} below 
uses $\delta$-sparsity in a subtle way to 
prepare for the proof that $\B_q$ contains $M_q$.
This result will also be used later to show that the graph induced complex $\G^\alpha(P,Q,d_E)$ in fact contains the restricted Delaunay triangulation $\resdel{Q}{M}$. 

For a simplex $\sigma\in \resdel{Q}{M}$, we call
a ball $B(c,r)$ a {\em surface Delaunay ball} of $\sigma$ if
$c\in M$ and its boundary contains the vertices of $\sigma$. 

\begin{figure*}[tbhp]
\begin{center}
\includegraphics[height = 3.8cm]{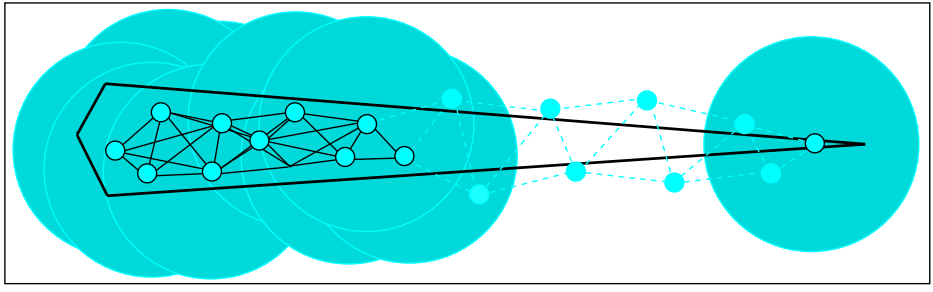}
\caption{In a long thin Voronoi cell, $\B_q$ may be
disconnected and may not contain
$M_q$.}
\label{fig:disconnect_skeleton}
\end{center}
\end{figure*}

\begin{proposition} \label{LEMMA:FIND-POINT}
Let $P$ be an $\eps$-sample of $(M,d_E)$, 
and $Q$ a $\delta$-sparse $\delta$-sample of $(P,d_E)$.
Let $\sigma\in \resdel{Q}{M}$ be a restricted Delaunay triangle or edge 
with a vertex $q \in Q$.
Let $c$ be the center of a surface Delaunay ball of $\sigma$.
If $8\eps \leq \delta \leq \frac{2}{27}\rho(M)$,
then there is a point $p\in P$ so that $p \in B(c, 4\eps)$ and 
$q$ is the closest point to $p$ among all points in $Q$.
\end{proposition}
\begin{proof}
See Appendix \ref{appendix:find-point} for the proof. 
\end{proof}

Now, we are ready to prove that $M_q$ is contained in the union
of balls $\bigcup_{P_q}B(p,\frac{\alpha}{2})$.

\begin{proposition}
If $\alpha \geq 12\eps$$ and 8\eps \leq \delta \leq \frac{2}{27}\rho(M)$, then
$M_q \subset \bigcup_{\{p\in P_q\}}B(p, \frac{\alpha}{2})$
which implies that $\Rips^{\alpha}(P_q)$ is path connected
if $M_q$ is path connected.
\label{Rips-connect}
\end{proposition}
\begin{proof}
Since $P$ is an $\eps$-sample of $M$,
$\forall x \in M$, there exists a point $p \in B(x, \eps)$ where $p \in P$.  
Let $P_q^{'} = P \cap (\bigcup_{x\in M_q}B(x, \eps))$. 
Then, we have $M_q \subset \bigcup_{p\in P_q^{'}}B(p, \eps)$ for if $x\in M_q$,
there exists $p\in P_q'$ with $p\in B(x,\eps)$ requiring $x\in B(p,\eps)$. 
On the other hand, recall that $P_q = M_q \cap P$. 
Hence if $p \in P_q^{'} \setminus P_q$, then $B(p, \eps)$ contains some boundary 
point $x\in \partial M_q$. The point $x$ belongs to a Voronoi facet
in the Voronoi diagram of $Q$ and hence 
$B(x, \|q-x\|)$ is a surface Delaunay ball. 
By Proposition~\ref{LEMMA:FIND-POINT}, 
we can find a point $u \in P_q$ such that $\|u-x\| \leq 4\eps$. 
Thus, $B(p,\eps) \subset B(u, (4+2)\eps)$. 
Taking $\alpha \geq 12\eps$, 
we get that $M_q \subset \bigcup_{p\in P_q^{'}}B(p, \eps) \subset \bigcup_{p\in P_q}B(p, \frac{\alpha}{2}) $. 

Since every ball in $\{B_p|p\in P_q\}$ intersects $M_q$,
we have that
$\B_q = \bigcup_{\{p\in P_q\}}B(p, \frac{\alpha}{2})$ is path connected
if $M_q$ is path connected.
On the other hand $\Rips^{\alpha}(P_q)$ is path connected if
$\B_q$ is path connected proving the claim.
\end{proof}

We can now present the main result of this subsection.

\begin{theorem}

Let $P$ be an $\eps$-sample of a smooth compact surface $M$ embedded in $\reals^3$, and $Q \subseteq P$ a $\delta$-sparse $\delta$-sample of $(P, d_E)$. 
For $12\eps \leq \alpha \leq \frac{2}{27}\rho$
and $8\eps \leq \delta \leq \frac{2}{27}\rho$, the map 
$h_*: \homo_1(\Rips^{\alpha}(P))\rightarrow \homo_1(\G^{\alpha}(P,Q,d_E))$ 
is an isomorphism where $h:\Rips^{\alpha}(P)\rightarrow \G^{\alpha}(P,Q,d_E)$
is the simplicial map induced by the nearest point map $\nu_{d_E}: P\rightarrow Q$.
\label{Surface-iso}
\end{theorem}
\begin{proof}
Since $\delta \leq 0.18\rho$, we can assume each restricted Voronoi
cell $M_{q}$ to be path connected~\cite{Dey}. This together with
the lower bound on $\alpha$ imply that $\Rips^{\alpha}(P_q)$
is connected for each $q\in Q$ thanks to Proposition~\ref{Rips-connect}.
Consequently, Proposition~\ref{surjectivity} 
establishes that $h_*$ is surjective. 

From Proposition 4.1 of \cite{DW11} and its proof, we obtain the following: 
for any $4 \eps \le r \le 2r \le r' \le \sqrt{\frac{3}{5}}\rho$, 
\begin{align}
\homo_1(\Rips^{r}(P)) &\cong \homo_1(\Rips^{r'}(P)) \cong \homo_1(M) 
\label{eqn:inclusioniso}
\end{align}
where the first isomorphism is induced by the canonical inclusion $i: \Rips^{r}(P) \hookrightarrow \Rips^{r'}(P)$. 
Our assumption on the ranges of
$\alpha$ and $\delta$ implies
the required conditions 
that $4\eps \leq \alpha \leq \frac{1}{3}\sqrt{\frac{3}{5}}\rho$ and $4\eps \leq \delta \leq \frac{1}{3}\sqrt{\frac{3}{5}}\rho$.
We claim that $i_{*}:\homo_1(\Rips^{\alpha}(P)) \rightarrow \homo_1(\Rips^{\alpha + 2\delta}(P))$ induced by the inclusion $i: \Rips^{\alpha}(P) \rightarrow \Rips^{\alpha + 2\delta}(P)$  is an isomorphism.

First, note that this claim follows easily from Eqn (\ref{eqn:inclusioniso}) if $\alpha \le \delta$ by setting $r = \alpha$ and $r' = \alpha + 2\delta (\ge 2 r)$. 
Now assume that $\delta \le \alpha$. 
Consider the following sequence: 
\begin{align*}
\Rips^{\delta}(P) \stackrel{i_1}{\hookrightarrow} \Rips^{\alpha}(P) 
\stackrel{i}{\hookrightarrow} \Rips^{\alpha + 2\delta}(P) 
\stackrel{i_2}{\hookrightarrow} \Rips^{3\alpha}(P). 
\end{align*}
By Eqn (\ref{eqn:inclusioniso}), we have that the composition of inclusions $ i \circ i_1: \Rips^{\delta}(P) \rightarrow \Rips^{\alpha + 2\delta}(P)$ induces an isomorphism at the homology level. Hence $i_*$ is necessarily surjective. 
On the other hand, the composition of inclusions $i_2 \circ i:  \Rips^{\alpha}(P) \rightarrow  \Rips^{3\alpha}(P)$ induces an isomorphism at the homology level. Hence $i_*$ is necessarily injective. 
Putting these two together, we have that $i_*$ is indeed an isomorphism. 
Therefore $h_*$ is injective by Proposition \ref{prop:contg}. 
It then follows that $h_*$ is an isomorphism as claimed. 
\end{proof}

Notice that the lower bound on $\delta$ in Theorem~\ref{Surface-iso} is not
restricted by $\alpha$. This means that one can have a dense
input graph for a large $\alpha$ whose connectivity
does not restrict the size of the subsample.

In the next subsection, we show two examples
of surface data where the graph induced complex has 
the correct $\homo_1$-homology 
with a considerably fewer simplices than the $\nu$-witness
complex, a modified witness complex suggested in~\cite{BGO09} for 
capturing the topology correctly.

\subsection{Reconstruction of $M$ using $\G^{\alpha}(P,Q,d_E)$}
In this subsection, we observe that the graph induced complexes
can also be used for surface reconstruction. It is known
that if $P$ is dense and $T$ is a simplicial complex with vertex set
$P$ which satisfies the following conditions, a simplicial manifold
can be extracted from $T$ that is homeomorphic to
$M$~\cite{ACDL02,Dey}. The conditions are: 
(i) $T$ is embedded in $\mathbb{R}^3$,
(ii) all triangles in $T$ have small circumradius with respect to reach
and (iii) $T$ contains the restricted Delaunay triangulation.
We show that $\G^{\alpha}(P,Q,d_E)$ 
contains the restricted Delaunay triangulation.
We then prune $\G^{\alpha}(P,Q,d_E)$ so that conditions (i)-(ii)
are satisfied, but none of the restricted Delaunay
triangles are deleted in the process which then ensures condition (iii). 

\begin{theorem} \label{thm:DelInclusion}
For $8\eps \leq \delta \leq \frac{2}{27}\rho$ and $\alpha \geq 8\eps$, 
we have that $\resdel{Q}{M}\subseteq \G^{\alpha}(P,Q,d_E)$ where
$P$ is an $\eps$-sample of $(M,d_E)$ and $Q\subseteq P$ is a 
$\delta$-sparse $\delta$-sample
of $(P,d_E)$.
\end{theorem} 
\begin{proof}
We will show that if $8\eps \leq \delta \leq \frac{2}{27}\rho$
and $\alpha \geq 8\eps$,  then any triangle $\sigma \in \resdel{Q}{M}$ is in $\G^{\alpha}(P,Q,d_E)$. 
The theorem follows from this. 

Let $\sigma=q_1q_2q_3$, and $c$ the center of a surface Delaunay ball of $\sigma$. 
By Proposition~\ref{LEMMA:FIND-POINT}, 
there exists a point $p_i\in P$ in $B(c, 4\eps)$ so that 
$q_i$ is the closest point in $Q$ to $p_i$ for $i=1,2,3$. 
It turns out that the interior of bounded cones used in 
the proof of Proposition~\ref{LEMMA:FIND-POINT} 
for $q_1$, $q_2$ and $q_3$ are disjoint.  
Hence each point $p_i$ found in $B(c, 4\eps)$ corresponding 
to $q_i$ is distinct from the other two. 
Therefore, if $\alpha \geq 8\eps$, the vertices
$p_1,p_2$ and $p_3$ form a clique in $G^{\alpha}(P)$ and
hence the triangle $q_{1}q_{2}q_{3}$ is in $\G^{\alpha}(P,Q,d_E)$. 
\end{proof} 

The complex $\G^{\alpha}(P,Q,d_E)$ may have intersecting triangles. 
We prune $\G^{\alpha}(P,Q,d_E)$ to eliminate
all such pairwise intersections while leaving the restricted
Delaunay triangles in the complex. This ensures that the resulting
complex embeds in $\mathbb{R}^3$ and still contains the restricted Delaunay
triangulation. Our simple observation is that if two intersecting triangles
$t_1$ and $t_2$ do not intersect in a common face, one can decide \emph{locally} 
which of the two can possibly be in a Delaunay triangulation. 
\begin{observation}
If $V$ is the vertex
set of two intersecting triangles
$t_1$ and $t_2$ whose intersection is not a common face of both,
then at least one of $t_1$ and $t_2$ 
is not in $\del{V}$. The triangle which is not in
$\del V$ cannot be in $\del P$ where $V\subseteq P$. 
\end{observation}
One can check locally the Delaunay condition for $t_1$ and $t_2$
and decide to throw away at least one triangle which is not in $\del{V}$. 
This takes only constant time since $V$ contains at most $6$ vertices. 
Notice that no restricted Delaunay triangle can be thrown away by this
process. After repeatedly pruning away one of the pairwise intersecting
triangles, we arrive at a complex that embeds in $\mathbb{R}^3$ and contains 
the restricted Delaunay triangulation $\resdel{Q}{M}$. Next, we prune all
triangles that have circumradius more than $2\delta$. 
Again, since the surface Delaunay ball of each restricted 
Delaunay triangle has circumradius at most 
$\delta+\eps\leq 2\delta$, 
one is ensured that no restricted Delaunay triangle
is eliminated. Assuming $\delta$ to be sufficiently small,
a sharp edge pruning and a walk on the outside
of the resulting complex as described in~\cite{ACDL02,Dey}
provides the reconstructed surface. 
The output surface enjoys one nice property that the triangles have
bounded aspect ratios since they have
circumradii of at most $2\delta$ and 
edge lengths of at least $\delta$ 
($Q$ is $\delta$-sparse).

\begin{theorem}
If $8\eps \leq \delta \leq \frac{2}{27}\rho$, $\alpha \geq 8\eps$, 
$P$ is an $\eps$-sample of $(M,d_E)$, and $Q\subseteq P$ is a 
$\delta$-sparse $\delta$-sample of $(P,d_E)$, then a triangulation
$T\subseteq \G^{\alpha}(P,Q,d_E)$ of $M$ can be computed where 
each triangle in $T$ has a bounded aspect ratio.
\end{theorem}

Experimentally, we observe that surfaces can be reconstructed
from a very sparse subsample with this strategy.
Figure (\ref{fig:fe_triangulation}) presents two examples
for surface reconstruction.
The original sample $P$ has $1,575,055$ points for the
Fertility model and $1,049,892$ points for Botijo model.
The input graphs for the graph induced complex
are constructed by connecting two points within
distance of $\alpha = 0.45$ for {\sc Fertility }
and $\alpha = 1.0$ for {\sc Botijo}.
The 2-skeleton of the Rips complex built on the input graph
has $45,788,607$ simplices for {\sc Fertility} and
$91,264,091$ simplices for {\sc Botijo}.
The subsample $Q$ consists of $3007$ points for
{\sc Fertility}  with $\delta = 3.68$,
and $4659$ points for {\sc Botijo}  with $\delta = 4.0$.
The graph induced complex $\G^{\alpha}(P,Q,d_E)$ built
on the subsample
has : $3007$ vertices, $9178$ edges, $6304$ triangles, $139$ tetrahedra
and no other higher dimension simplices for {\sc Fertility};
$4659$ vertices, $14709$ edges, $10755$ triangles, $718$ tetrahedra,
$5$ $4$-dimensional simplices,
 and no other higher dimension simplices for {\sc Botijo}.
The reconstructed surfaces from $\G^{\alpha}(P,Q,d_E)$
are shown in Figure~\ref{fig:fe_triangulation}.
For {\sc Fertility}, it has $3007$ vertices, $9039$ edges and $6026$ triangles;
for {\sc Botijo}, it has $4659$ vertices, $14001$ edges and $9334$ triangles.
Evidently, the graph induced complex
has only a few more simplices compared to the reconstructed surface.

For a comparison, we also constructed
the $\nu$-witness complex suggested in~\cite{BGO09} which
also contains the restricted
Delaunay triangulation $\resdel{Q}{M}$ with $\nu=(1, 6, 6, 4)$.
The $\nu$-witness complex for {\sc Fertility} has $3007$ vertices,
$35687$ edges, $119237$ triangles and $19874$ tetrahedra;
the $\nu$-witness complex for {\sc Botijo} has
$4659$ vertices, $54648$ edges, $180936$ triangles and $29654$ tetrahedra.
The graph induced complex has much smaller size,
but still captures $\beta_1$
($\beta_1 = 8$ for the {\sc Fertility}, and $\beta_1 = 10$
for the {\sc Botijo}).

\begin{figure*}[htpb]
\begin{center}
\begin{tabular}{cc}
\includegraphics[scale = 0.4]{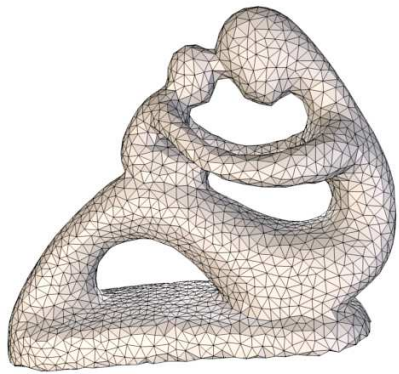}
&
\includegraphics[scale = 0.4]{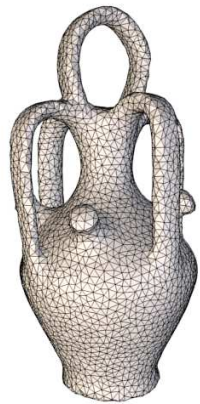} \\
(a) {\sc Fertility} model & (b) {\sc Botijo} model
\end{tabular}
\caption{Reconstructed surfaces for {\sc Fertility}
 and {\sc Botijo} models.}
\label{fig:fe_triangulation}
\end{center}
\end{figure*}

\section{Point data for more general domains}
In this section we consider domains beyond surfaces in $\mathbb{R}^3$.
\subsection{Manifolds}
Let $M$ be a $k$-manifold embedded in
$\mathbb{R}^n$ and $P$ a discrete sample of $(M,d_E)$. 
We observe that the overall setup in section~\ref{surface-topoinfer}
for inferring $\homo_1$-homology
from the graph induced complex generalizes easily to higher dimensions.
The inclusion map $\Rips^{\alpha}(P) \hookrightarrow  \Rips^{\alpha+2\delta}(P)$
still induces an isomorphism at the homology level if 
$\alpha$ and $\delta$ are chosen appropriately. In that case, the
map $h_*:\homo_1(\Rips^{\alpha}(P)) \rightarrow \homo_1(\G^{\alpha}(P,Q,d))$
remains injective by the same argument as before.
The main trouble arises when we try to prove that it is also surjective.
Observe that, to prove the surjectivity of $h_*$, we used the fact that
the restricted Voronoi cell $M_q=V_q\cap M$ in $\vor Q$ is
connected (Proposition~\ref{Rips-connect}). Unfortunately, 
this is not necessarily
true in high dimensions given the counterexamples in~\cite{BGO09,CDR05}.
To overcome this impediment we change the distance function replacing the
Euclidean distance with the graph distance while building
$\G^{\alpha}(P,Q,d)$. Specifically, we still consider $G^{\alpha}(P)$ to be the 
graph connecting points in $P$ with Euclidean
distance $\alpha$ or less, but take $Q$ to be
a $\delta$-sparse $\delta$-sample of $(P,d_G)$ where the graph distance
$d_G=d_{G^{\alpha}(P)}$ is defined with 
the Euclidean lengths as the edge weights.
Then, we consider $\G^{\alpha}(P,Q,d_G)$.

As before, let $P_q\subseteq P$ be the set of points nearest to a point
$q\in Q$ with respect to $d_G$. 
The modification in distance function immediately allows us to
claim that the $1$-skeleton of $\Rips^{\alpha}(P_q)$ is connected,
which was needed to claim that $h_*$ is surjective. 

\begin{proposition}
$\Rips^{\alpha}(P_q)$ is connected, and thus $h_*$ is surjective. 
\label{rpq-connected}
\end{proposition} 

\begin{theorem}
Let $P$ be an $\eps$-sample of an embedded smooth and compact manifold $M$ with reach $\rho$, 
and $Q$ a $\delta$-sample of $(P, d_G)$.  
For $4\eps \le \alpha, \delta \le \frac{1}{3} \sqrt{\frac{3}{5}} \rho$, 
 the map $h_*: \homo_1(\Rips^{\alpha}(P))\rightarrow \homo_1(\G^{\alpha}(P,Q,d_G))$ 
is an isomorphism where $h:\Rips^{\alpha}(P)\rightarrow \G^{\alpha}(P,Q,d_G)$
is the simplicial map induced by the nearest point map $\nu_{d_G}: P\rightarrow Q$.
\label{cor:manifoldgeodesic}
\end{theorem}

\subsection{A leaner subsampling for $\homo_1$}
In this subsection we show that the subsample $Q$ can be made leaner.
The main insight is that
we can define a feature size larger than the reach which permits
us to subsample more sparsely
with respect to this larger feature size. 
Gao et al.~\cite{GGOW10} considered a similar feature size for the
same reason of requiring sparser sampling for a two dimensional shape.
Here we show that such a sparser sample is also adequate for determining
$\homo_1$ of manifolds in high dimensions. Our experimental results
in Figure~\ref{fig:comparison} 
suggest that one can obtain information about $\homo_1$ from a very sparse 
sample in practice.

Let $\mathcal{K}$ be a simplicial complex 
with non-negative weights on its edges.
We define \textit{homological loop feature size} as
\[\mathrm{hlfs}(\mathcal{K}) = \left\{ \begin{array}{ll}
\frac{1}{2} \inf \{|c|, \mbox{ $c$ is
non null-homologous 1-cycle in } \mathcal{K}\}\\
\infty \mbox{ if no such $c$ exists.}
\end{array}
\right. \]
This feature size is very similar to the 
\textit{systolic feature size } $\mathrm{sfs}(X,d)$ of a 
compact metric space $(X,d)$ \cite{GGOW10} 
which is the length of the
shortest non-contractible loop in $X$.
Our definition of $\mathrm{hlfs}$
when applied to a metric space $(X,d)$
becomes larger than or equal to $\mathrm{sfs}(X,d)$ .
Notice that every loop of $\mathcal{K}$ with length 
less than $2\mathrm{hlfs}(\mathcal{K})$ is null-homologous in $\mathcal K$. 

Let $Q\subseteq P$ be a $\delta$-sample of $(P,d_G)$
as before but with $\delta\leq\frac{1}{2}\mathrm{hlfs}(\mathcal{R}^{\alpha}(P)) - \frac{1}{2}\alpha$.
Let $h: \mathcal{R}^{\alpha}(P) \rightarrow \G^{\alpha}(P,Q,d_G)$ be the 
simplicial map as defined earlier. We aim to show that the 
induced homomorphism $h_{*}$ on 
the first homology is injective.
Since we use graph distances, Proposition~\ref{rpq-connected} 
remains valid and hence
$h_*$ remains to be surjective. However, we cannot claim
$i_*: \homo_1(\Rips^{\alpha}(P))\rightarrow \homo_1(\Rips^{\alpha+2\delta}(P))$
to be an isomorphism because $\delta$ could be larger than required.
Thus, we cannot use $i_*$ to infer $h_*$ to be injective
as before. Nevertheless, we can prove the following result using a different
approach.

\begin{theorem}
\label{injectivity theorem}
If $Q$ is a $\delta$-sample of $(P,d_{G})$
for $\delta < \frac{1}{2}\mathrm{hlfs}(\mathcal{R}^{\alpha}(P)) - \frac{1}{2}\alpha$,
then $h_{*}: \homo_{1}(\mathcal{R}^{\alpha}(P)) \rightarrow
\homo_{1}(\G^{\alpha}(P, Q,d_G))$ is an isomorphism. 
\label{thm:injectivityforsys}
\end{theorem}
\begin{figure}[htbp]
\begin{center}
\input{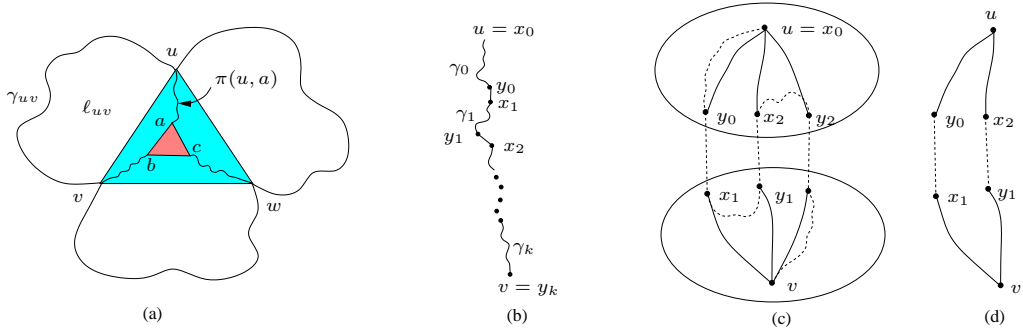}
\caption{(a) $\gamma_{uv}$ makes a cycle with $\pi(u,a)$, $\pi(v,b)$ and $ab$,
(b) $\gamma_{uv}$ as a sum of unicolored chains and bicolored edges, (c) converting $\gamma_{uv}$ (shown dotted) to $\hat{\gamma}_{uv}$, (d) a diamond of
$\hat{\gamma}_{uv}$.
}
\label{hlfs-fig}
\end{center}
\end{figure}
%
\begin{proof}
We only need to show that $h_*$ is injective, as its surjectivity follows from Proposition \ref{rpq-connected}.
To show the injectivity, it suffices to show that $h_{*}$ has a trivial kernel.
Let $\sigma$ be any triangle in $\G^{\alpha}(P, Q,d_G)$.
If under the chain map $h_{\#}$, every cycle in the preimage of the boundary
cycle $\partial \sigma$ is null homologous,
then every null homologous cycle in $\G^{\alpha}(P, Q,d_G)$
has only null homologous cycles in its preimage.
This is true due to the fact that a bounded cycle is a sum of boundaries
of triangles, and that the chain map $h_{\#}$ is surjective (see the proof of Proposition \ref{surjectivity}).
Below we show that under the chain map $h_{\#}$,
every cycle in the preimage of the boundary
cycle of any triangle indeed is null homologous.
It would then follow that the kernel of $h_*$ is trivial.

Let $\gamma$ be any cycle in the preimage of $\partial \sigma$.
We have $\gamma\in\sum_{uv\in\partial\sigma} h_{\#}^{-1}(uv)$
where $uv$ be any edge of of $\sigma$.
Let $\gamma_{uv}$ be any maximal subpath of $\gamma$ so that
$h_{\#}(\gamma_{uv})=uv$ (Figure~\ref{hlfs-fig}(a)).
For each such $\gamma_{uv}$, we construct a cycle $\ell_{uv}$ so that
$\ell_{uv}$ is null homologous and $\gamma$ is homologous
to $\sum\ell_{uv}$. Therefore,
our problem of showing $\gamma$ is null homologous
reduces to the problem of showing every $\ell_{uv}$ is null-homologous.

We construct $\ell_{uv}$ as follows.
By the construction of $\G^{\alpha}(P, Q,d_G)$,
there is a triangle $abc \in \mathcal{R}^\alpha(P)$
such that $h(abc) = \sigma$ with $h(a) =u$, $h(b)=v$.
Consider the shortest paths $\pi(u,a)$ and
$\pi(v,b)$ in $G^{\alpha}(P)$ from $u$ to $a$ and from $v$ to $b$
respectively. Observe that all vertices in $\pi(u,a)$ and $\pi(v,b)$
are mapped to $u$ and $v$ respectively by $h$ since we are using the graph-induced distance $d_G$ to construct $\G^{\alpha}(P, Q,d_G)$.
Take $\ell_{uv}$ to be the chain $\pi(u,a)+ ab +\pi(v,b) + \gamma_{uv}$;
refer to Figure~\ref{hlfs-fig}(a).
With this choice, we have $\gamma=\sum \ell_{uv} + \partial(abc)$
and hence $\gamma$ is homologous to $\sum\ell_{uv}$ as promised.
To prove $\ell_{uv}$ null-homologous, we construct a homologous
path $\hat{\gamma}_{uv}$ to $\gamma_{uv}$ which gives a homologous
cycle $\hat{\ell}_{uv}$ to $\ell_{uv}$. We then prove that $\hat{\ell}_{uv}$
is null-homologous.

Call an edge $e = (x,y)$ in $\mathcal{R}^\alpha(P)$ \emph{bicolored} if its
two end-points are mapped to two distinct vertices by $h$; otherwise,
$e$ is \emph{unicolored}. A $1$-chain from $\mathcal{R}^\alpha(P)$ is
\emph{unicolored} if it has only unicolored edges.
For simplicity, we assume that vertices from $\gamma_{uv}$ are all contained in $h^{-1}(u) \cup h^{-1}(v)$ because $\gamma_{uv}$ is always
homologous to a path containing vertices mapped only to $u$ or $v$.
In this case, $\gamma_{uv}$ can be decomposed into a set of bicolored
edges $\{ y_0x_1, y_1x_2, \ldots, y_{k-1}x_k \}$ together with a set of
unicolored chains $\{ \gamma_0, \gamma_1, \ldots, \gamma_k \}$
such that $\partial \gamma_i = x_i + y_i$.
In particular, $\gamma_{uv}$ can be written as
\begin{equation}
\gamma_{uv} = \gamma_0 + y_0x_1 + \gamma_1 + y_1x_2 + \cdots + y_{k-1}x_k+\gamma_k.
\label{eqn:gamma}
\end{equation}
See Figure \ref{hlfs-fig}(b) for an illustration.

\begin{claim} Let $\gamma_i$ be a unicolored chain with two boundary points
$x_i, y_i$ so that for any simplex $\tau \in \gamma_i$, $h(\tau)=u$. Then, $\gamma_i$ is homologous to the
chain $\hat{\gamma}_i=\pi(x_i,u)+\pi(u,y_i)$.
\label{path-reduce}
\end{claim}

Given the chain $\gamma_{uv}$ with subchains $\gamma_i$
as in equation~\ref{eqn:gamma},
we convert it to a homologous chain
\[
\hat{\gamma}_{uv}=\hat{\gamma_0} + y_0x_1+\hat{\gamma_1}
+\cdots + y_{k-1}x_k+\hat{\gamma}_k .
\]
Replace $\gamma_{uv}$ with $\hat{\gamma}_{uv}$
in $\ell_{uv}$ to obtain a homologous cycle
$\hat{\ell}_{uv}$; refer to Figure~\ref{hlfs-fig}.
Observe that $\hat{\ell}_{uv}$
is the sum of cycles (diamonds) that have two unicolored
chains and two bicolored edges as shown in Figure~\ref{hlfs-fig}(d). Such
a cycle $c$ has length at most $4\delta +2\alpha$.  This is because
each unicolored chain in $c$ has at most two shortest paths
of the form $\pi(u,x_i)$ and $\pi(u,y_i)$ (or $\pi(v,x_i)$ and $\pi(v,y_i)$)
that have lengths $2\delta$
or less ($Q$ is a $\delta$-sample of $(P,d_G)$), and the two bicolored
edges have lengths at most $2\alpha$ in $G^{\alpha}(P)$. The cycle
$c$ is null homologous because its length is
$$|c| \le 4 \delta + 2\alpha < 2 \mathrm{hlfs}(\mathcal{R}^{\alpha}(P)), ~~~~\text{given that}~~\delta < \frac{1}{2}\mathrm{hlfs}(\mathcal{R}^{\alpha}(P)) - \frac{1}{2}\alpha. $$
It follows that $\hat{\ell}_{uv}$ is null
homologous as we are required to prove.

We only need to show Claim~\ref{path-reduce} to finish the proof. Let
$x_i=p_0,p_1,\cdots,p_m=y_i$ be the sequence of vertices
on the unicolored chain (path) $\gamma_i$. Consider the shortest
paths $\pi(p_i,u)$ for each $p_i$ on this path. The
length of the cycle $z_i=\pi(u,p_i)+p_ip_{i+1}+\pi(u,p_{i+1})$ is
at most $2\delta +\alpha$ for each $i\in [0,m-1]$. Therefore,
it is null homologous by our assumption. We have
$\gamma_i+\hat{\gamma_i}= \sum_{i-0}^{m-1} z_i=0$.
Therefore, $\gamma_i$ and $\hat{\gamma_i}$ are homologous.
\end{proof}

\subsection{Point data for compact sets}
\label{compact-sec}
So far we have focused on $\homo_1$-homology. In this section
we extend the domain to compact subspaces of Euclidean spaces and
consider homology groups of all dimensions. This generality comes
at the expense of additional computations. Unlike previous approaches
that allow us to infer the $\homo_1$-homology of the
sampled manifold by computing directly the same for the graph
induced complexes, now we need to compute the persistent homology~\cite{EH09}
induced by simplicial maps. The well-known algorithms for computing
persistent homology~\cite{EH09} work for maps induced
by inclusions. In a contemporary paper~\cite{DFW12}, we present 
an algorithm that can compute the persistent homology induced by
simplicial maps.

Let $X\subset \mathbb{R}^n$ be a compact set and $X^{\lambda}$ be its 
offset with $\lambda > 0 $. Since it is difficult to
compute $\homo_k(X)$ from a sample~\cite{CO08},
we aim for computing
the homology groups $\homo_k(X^{\lambda})$ for the offset $X^{\lambda}$.
Let $\mathrm{wfs}(X)$ denote the 
the\textit{ weak feature size} which is defined as
the smallest positive critical value of the distance function
to $X$~\cite{CL05}. 
We prove that the persistent homology of the graph 
induced complex defined with either Euclidean or graph distance $d$ provides 
the correct homology of $X^{\lambda}$ where $0 < \lambda < \mathrm{wfs}(X)$.
Specifically, the image of 
$h_*:\homo_k(\G^{\alpha}(P,Q,d))\rightarrow \homo_k(\G^{\alpha'}(P,Q',d))$ 
induced by a simplicial map $h:\G^{\alpha}(P,Q,d)\rightarrow \G^{\alpha'}(P,Q',d)$
becomes isomorphic to $\homo_k(X^{\lambda})$ for
appropriate $\alpha$ and $\alpha'$.
We recall the following result from \cite{CO08}. 
\begin{proposition}
\label{lemma:six-seq}
If the sequence of homomorphisms $A\rightarrow B \rightarrow C \rightarrow D \rightarrow E \rightarrow F$ between
finite dimensional vector spaces satisfies that
 $\mathrm{rank}(A\rightarrow F) = \mathrm{rank}(C\rightarrow D)$,
then $\mathrm{rank}(B\rightarrow E) = \mathrm{rank}(C\rightarrow D)$.
\end{proposition}

Let $Q$ and $Q'$ be subsamples of $P$ where
$Q$ is a $\delta$-sparse $\delta$-sample and 
$Q'$ is a $\delta'$-sparse $\delta^{'}$-sample for $\delta^{'} > \delta$.
Consider the interleaving sequence between the 
graph induced and Rips complexes,
\begin{equation} \label{diag:gic-rips}
\xymatrix
{
\Rips^{\alpha}(P) \ar@{^{(}->}[r]^-{i_1} \ar@{->}[rd]^-{h_1}  
& \  \Rips^{\alpha+2\delta}(P) \  \ar@{^{(}->}[r]^-{i_2} 
& \ \Rips^{4(\alpha+2\delta)}(P) \  \ar@{^{(}->}[r]^-{i_3} \ar@{->}[d]^-{h_2} 
& \ \Rips^{4(\alpha+2\delta)+2\delta^{'}}(P)  
\\
& \G^{\alpha}(P,Q,d)  \ar@{.>}[r]^-{h} \ar@{^{(}->}[u]^-{j_1}
& \G^{4(\alpha+2\delta)}(P,Q',d) \   \ar@{^{(}->}[ru]^-{j_2}  
&
}
\end{equation}
where $i_1,i_2,i_3,j_1$ and $j_2$ are inclusions and $h=h_2\circ i_2\circ j_1$.
By Proposition~\ref{prop:contg}, $h_1$ and $h_2$ are simplicial maps .
Therefore, $h$ is also a simplicial map as composition of simplicial maps.
In particular, $h$ is the simplicial map induced by the vertex map that 
maps each point $q \in Q$ to its closest point $q'\in Q'$ in $Q'$. 
We prove that $\mathrm{im}\ h_* \cong \homo_{k}(X^{\lambda})$ where 
$h_{*} : \homo_{k}(\G^{\alpha}(P,Q,d)) \rightarrow \homo_{k}(\G^{4(\alpha+2\delta)}(P,Q',d))$
and $\eps$ , $\alpha$ and  $\delta$ fall in appropriate ranges.

\begin{theorem}
Let $X\subset \mathbb{R}^{n}$ be a compact space. 
Let $0 < \eps < \frac{1}{9}\mathrm{wfs}(X)$ and $P$ 
be an $\eps$-sample of $(X,d_E)$.
Let $Q$ be a $\delta$-sparse $\delta$-sample of $(P,d)$ 
and $Q'$ be a 
$\delta'$-sparse $\delta^{'}$-sample of $(P,d)$ 
where $d$ is either Euclidean or graph distance and
$\delta^{'} > \delta$.\\

If $2\eps \leq \alpha \leq \frac{1}{4}(\mathrm{wfs}(X)-\eps)$ and
$(\alpha + 2\delta)+\frac{1}{2}{\delta}^{'} \leq \frac{1}{4}(\mathrm{wfs}(X)-\eps)$, then 
$\mathrm{im}\ h_* \cong \homo_{k}(X^{\lambda})$ $(0 < \lambda < \mathrm{wfs}(X))$ where 
$h_{*} : \homo_{k}(\G^{\alpha}(P,Q,d)) \rightarrow \homo_{k}(\G^{4(\alpha+2\delta)}(P,Q',d))$ is induced by $h$ in diagram~\ref{diag:gic-rips}.
\end{theorem}
\begin{proof}
The diagram \ref{diag:gic-rips} is not commutative in general.
However, it is commutative at the homology level.
Proposition~\ref{prop:contg} makes the two triangles at the left and
right commutative. The middle square commutes by definition of $h$.
Now consider the sequence,
\begin{eqnarray}
\xymatrix 
{
\homo_{k}(\Rips^{\alpha}(P)) \ar@{->}[r]^-{{h_1}_{*}} 
& \homo_{k}(\G^{\alpha}(P,Q,d))  \ar@{->}[r]^-{{j_1}_{*}}
& \  \homo_{k}(\Rips^{\alpha+2\delta}(P)) 
\\
\  \ar@{->}[r]^-{{i_2}_{*}}   & \ \homo_{k}(\Rips^{4(\alpha+2\delta)}(P)) \  \ar@{->}[r]^-{{h_2}_{*}} 
& \homo_{k}(\G^{4(\alpha+2\delta)}(P,Q',d)) \   \ar@{->}[r]^-{{j_2}_{*}} 
& \ \homo_{k}(\Rips^{4(\alpha+2\delta)+2\delta^{'}}(P))  
}
\label{seq}
\end{eqnarray}
Consider the sequence of inclusions at the upper level of the diagram (1).
Since $\alpha \geq 2\eps$ and $(\alpha + 2\delta)+\frac{1}{2}{\delta}^{'} \leq \frac{1}{4}(\mathrm{wfs}(X)-\eps)$,
we have that 
$$\mathrm{im}\ (i_3\circ i_2\circ i_1)_{*} \cong \homo_{k}(X^{\lambda})
\mbox{ and }
\mathrm{im}\ (i_2)_{*} \cong \homo_{k}(X^{\lambda})$$ 
by Theorem $3.6$ of \cite{CO08}.
Considering the diagram in (1) and the sequence in (\ref{seq}) we have
\begin{eqnarray}
\homo_k(X^{\lambda})\cong \mathrm{im}\ ({i_3}_{*} \circ {i_2}_{*} \circ {i_1}_{*}) \cong \mathrm{im} (({j_2}_*\circ {h_2}_*) \circ {i_2}_* \circ ({j_1}_{*} \circ {h_1}_{*}))\cong 
\mathrm{im}\ {i_2}_{*}
\label{seq3}
\end{eqnarray}
Letting $h=h_2\circ i_2 \circ j_1$, the rightmost
isomorphism in (\ref{seq3}) allows us to 
claim that $\mathrm{im}\ h_* \cong \homo_{k}(X^{\lambda})$ by
applying Proposition~\ref{lemma:six-seq} to the sequence (\ref{seq}).
\end{proof}

\section{Conclusions}
In this work, we investigated the graph induced complex that 
can be built upon a given point cloud data
and a suitable graph connecting them.
This complex, to some extent, has the advantages of both Rips and witness
complexes. We have identified several of its topological
properties that can evidently 
be useful in extracting information from
point data even in high dimensions. 

In section~\ref{compact-sec}, we have shown how one can infer the homology
groups of a compact set using the persistent homology of a pair
of graph induced complexes constructed with two values of $\delta$.
One can consider a filtration of $\G^{\alpha}(P,Q,d)$ with $Q$ being
sparsified for increasing values of $\alpha$ and $\delta$. Then, one can obtain
a persistence diagram~\cite{CEH07} out of this 
`full filtration' using our recently
proposed algorithm for computing the topological persistence for 
filtrations connected with simplicial maps~\cite{DFW12}. The algorithm
will collapse vertices progressing through the filtration and hence
will keep the size of the complex in question contained. 
Relating this persistence diagram to that of a filtration obtained by
a related Rips filtration is an interesting
question. In a subsequent work,
we have addressed this question in~\cite{DFW12}. 

Finding other applications where the graph induced complex
becomes useful also remains open for further investigations.  

\section*{Acknowledgment} We acknowledge the support of the NSF grants
CCF 1116258 and CCF-1064416 and thank all the referees for their valuable
comments.

\vfill\eject

\newpage

\appendix

\section{Proof of Proposition \ref{LEMMA:FIND-POINT} }
\label{appendix:find-point}

\newcommand{\anangle}		{\alpha}
\newcommand{\mycone}		{\mathrm{C}}
\newcommand{\mys}		{\ell}

First, we present an elementary geometric result that we need to use in the proof. 
Let $\mycone(o, \vec{v}, \anangle)$ denote the cone with apex $o$, axis in the direction of $\vec{v}$ and aperture $2\anangle$. 
\begin{claim}
Given a ball $B$ with radius $r$ and center $c$, let $q$ be an arbitrary point on the surface of sphere. 
Consider the two nested cones $\mycone_1 = \mycone(c,\overrightarrow{cq}, \anangle)$ and $\mycone_2 = \mycone(c, \overrightarrow{cq}, 2\anangle)$ with the same axis. 
Let $p$ be any point from the intersection of the ball and the inner cone; that is, $p \in B \cap \mycone_1$. 
Let $x$ be an arbitrary point from the boundary of $B$ outside the outer cone; that is, $x \notin \mycone_2$, and $x \in \partial B$. 
Then we have that $\| p - q \| < \| p - x \|$. 
\label{claim:twocones}
\end{claim}
\begin{proof}
Denote $\anangle_q := \angle pcq$ and $\anangle_x :=\angle pcx$. 
Because $x$ is outside of the outer-cone with aperture $4\anangle$, and $p$ is inside of the inner cone of aperture $2\anangle$, we have that  
we have $\anangle_q < \anangle < \anangle_x$.
Now consider the triangle $\triangle pcq$. By the Cosine Law, we have that
\[
\|p - q \|^{2} = \|p - c \|^{2} + \|c - q\|^{2} - 2\|p - c\|\cdot\|c - q\| \cos(\anangle_d) =\|p - c \|^2 +  r^2  - 2r\|p-c\| \cos(\anangle_q). 
\]
Similarly, consider the $\triangle pox$, and we have 
\[
 \|p - x\|^{2} = \|p - c\|^{2} + \|c - x \|^{2} - 2\|p - c\|\cdot \|c - x\| \cos(\anangle_p) =\|p - c\|^2 +  r^2 - 2r\|p-c\|\cos(\anangle_x) .
\]
Since $0 \le \anangle_q < \anangle_x \le \pi$, we have $\|p - q\| < \|p - x \|$. 
\end{proof}

\vspace*{0.1in} Now consider the surface Delaunay ball $B_c = B(c,r)$ that passing through the vertices of the simplex $\sigma$ and containing no other points from $Q$. 
Recall that $q$ is an arbitrary vertex of $\sigma$. 
Since all other vertices of $\sigma$ are at least $\delta$-Euclidean distance away from $q$,  we then have that the intersection of $B_c$ with the cone $\mycone(c, \overrightarrow{cq}, 2\arcsin \frac{\delta}{2r})$ contains no point from $Q$ other than $q$. 
By applying Claim \ref{claim:twocones} with $\anangle = \arcsin \frac{\delta}{2r}$, we then obtain that:
\begin{cor}
If there exists a point $p \in P$ such that $p \in B_c \cap \mycone(c, \overrightarrow{cq}, \arcsin \frac{\delta}{2r}) \cap M$, then $q$ must be the closest point to $p$ among all points in $Q$. 
\label{cor:NN}
\end{cor}

In what follows, we will show that a point $p \in P$ satisfying the conditions in Corollary \ref{cor:NN} as well as that $p \in B(c, 4\eps)$ 
indeed exists when $8\eps \le \delta \le \frac{2}{27}\rho(M)$. This will then prove the proposition. 
Specifically, we will first identify a sample point $p\in P$, and then we will show that $p$ satisfies the requirements of the proposition. 

\parpic[r]{\includegraphics[height=4cm]{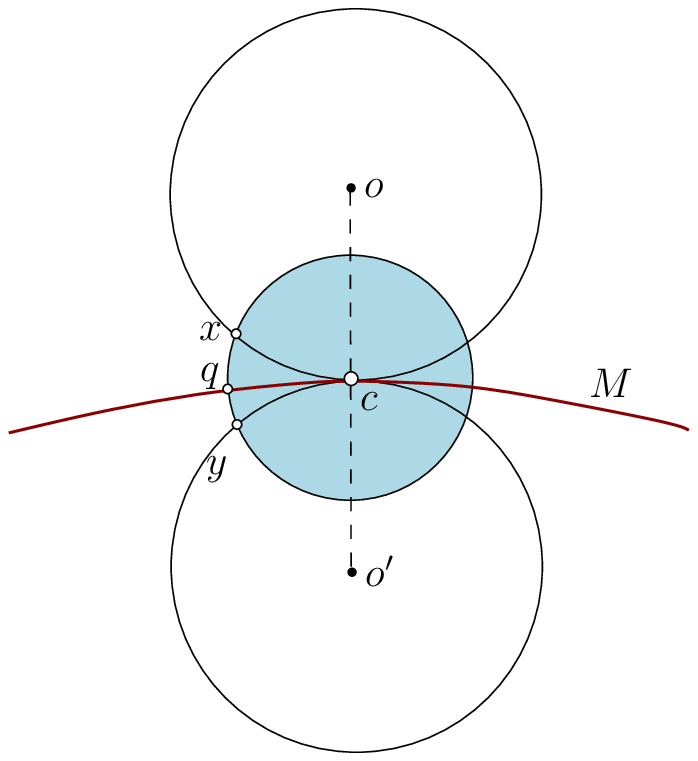}}
\vspace*{0.08in}\noindent{\bf Identifying a point $p \in P$.~}
Let $B_o=B(m,\rho)$ and $B_{o'}=B(m',\rho)$ be two balls tangent to $M$ at $c$; assume without loss of generality that $B_o$ is inside of $M$ and $B_{o'}$ is outside. Locally around $c$, the surface $M$ is sandwiched between $B_o$ and $B_{o'}$. 
Now consider the plane $\mathcal{P}= \mathrm{span}\{o, o', q\}$; note that $c$ also lies in $\mathcal{P}$. 
Denote $B_{c,\mathcal{P}} = B_{c} \cap \mathcal{P}$,
$B_{o, \mathcal{P}} = B_{o} \cap \mathcal{P}$ and $B_{o', \mathcal{P}} = B_{o'} \cap \mathcal{P}$. 
Let $x$ be the intersection point of $B_{c, \mathcal{P}}$ and $B_{o, \mathcal{P}}$ that is on the same side of the line passing through $oo'$ as the point $q$. Similarly, let $y$ be the intersection point of $B_{c, \mathcal{P}}$ with $B_{o', \mathcal{P}}$ on the same side of the line $oo'$ as $q$. 
Obviously, $q$ lines on the arc $\stackrel{\frown}{xy}$ that  
avoids $B_{o, \mathcal{P}}$ and $B_{o', \mathcal{P}}$. 
See the right figure for an illustration where the shaded region is $B_{c, \mathcal{P}}$. 
Set $\theta:= \angle xcy$; easy to see that $\theta = \angle xoc = \angle yo'c$. Hence we have that  
$\sin \frac{\theta}{2} = \frac{\|c - x \|}{2 \rho} = \frac{r}{2\rho}.$

Now consider the segment $cq$ and the point $w \in cq$ such that the length of $cw$ is a value $\mys$ which we will set later. 
As $w$ is contained in the cone with apex $c$ and aperture $\theta$ in the plane $\mathcal{P}$, the ball $B(w, \mys \sin \theta )$ will intersect both segment $cx$ and $cy$, thus intersecting both $B_o$ and $B_{o'}$. 
Since $B_o$ is inside the surface $M$ and $B_{o'}$ outside, it follows that $B(w, \mys \sin \theta) \cap M \neq \emptyset$. 
Pick any point $p' \in B(w, \mys \sin \theta) \cap M$. By the $\eps$-sampling condition of $P$, there must exist a sample point $p \in P$ such that $\| p - p'\| \le \eps$. 
In other words, there is a sample point $p \in P$ such that $p \in B(w, \mys \sin \theta + \eps)$. 

\paragraph{The requirements on $p$.}
We now need to show that the parameter $\mys$ can be chosen such that the point $p$ satisfies all the requirements from the proposition. 
In particular, we need the following: 
\begin{description}\denselist
\item[C-1] $p \in B(c, 4\eps)$; and 
\item[C-2] $q$ is the closest point to $p$ among all points in $Q$. 
\end{description}


Now set $\tau = \arcsin \frac{\mys\sin \theta + \eps}{\mys}$. Obviously, the ball $B(w, \mys \sin \theta + \eps)$ (and thus the point $p$) is contained inside the cone $\mycone(c, \overrightarrow{cq}, \tau)$. 
Observe that by Corollary \ref{cor:NN}, condition (C-2) is satisfied if (C-2.a) $B(w, \mys \sin \theta + \eps) \in B(c, r)$ (implying that $p \in M$), and (C-2.b) $\tau \le \arcsin \frac{\delta}{2r}$ (implying that $p$ is contained in the inner cone $\mycone(c, \overrightarrow{cq}, \arcsin \frac{\delta}{2r})$).  

\paragraph{The existence of a valid $\mys$.}
What remains is to find a value for $\mys$ so that (C-1), (C-2.a), and (C-2.b) are all satisfied simultaneously. 
Note that since $\| w - c \| = \mys$, we have that $\| p - c \| \le \mys + \mys \sin \theta + \eps$. 
Hence condition (C-1) is satisfied if $\mys + \mys \sin \theta + \eps \le 4\eps$. Since $\sin \theta \le 2 \sin \frac{\theta}{2} = \frac{r}{\rho}$, (C-1) holds as long as the following inequality holds. 
\begin{align}
\mys \le \frac{3\eps}{1 + \frac{r}{\rho}}. 
\label{eqn:R1}
\end{align}

Since $\delta \ge 8 \eps$, if (C-1) holds, then we have that $\| p - c \| \le 4\eps \le \frac{\delta}{2} \le r$, which implies (C-2.a). 
Now consider condition (C-2.b), which holds if $\frac{\mys \sin \theta + \eps}{\mys} \le \frac{\delta}{2r}$. 
Since $\delta/2 \le r \le \delta + \eps$ and $8\eps \le \delta < 2\rho / 27$, we have that: 
$$
\frac{r}{\rho} \le \frac{\delta+\eps}{\rho} \le \frac{\delta + \delta/8}{\rho} < \frac{1}{4} \le \frac{\delta}{4\delta} \le \frac{\delta}{2r}. 
$$
That is, $\frac{\delta}{2r} - \sin \theta \ge \frac{\delta}{2r} - \frac{r}{\rho} > 0$ (recall that $\sin \theta \le 2 \sin \theta = r / \rho$). 
Hence condition (C-2.b) holds if 
\begin{align}
\mys \ge \frac{\eps}{\frac{\delta}{2r}-\frac{r}{\rho}} (\ge \frac{\eps}{\frac{\delta}{2r} - \sin \theta}). 
\label{eqn:R2}
\end{align}

Putting Eqns (\ref{eqn:R1}) and (\ref{eqn:R2}) together, we have that as long as the value $\mys$ satisfying the following inequality:
\begin{eqnarray}
\frac{\eps}{\frac{\delta}{2r}-\frac{r}{\rho}} \le \mys \le \frac{3\eps}{1 + \frac{r}{\rho}}. 
\label{eqn:combined}
\end{eqnarray}
 then conditions (C-1) and (C-2) will be satisfied, and there exists a point $p \in P$ as stated in the proposition. 
Given that $8 \eps \le \delta < 2\rho / 27$, we can show that valid $\mys$ exists. For example, for $\mys = \frac{36\eps}{13}$, inequality in Eqn (\ref{eqn:combined}) holds as
\begin{align*}
\frac{\eps}{\frac{\delta}{2r}-\frac{r}{\rho}} \le \frac{\eps}{\frac{\delta}{2(\delta+\eps)} - \frac{\delta+\eps}{\rho}} \le \frac{\eps}{\frac{4\delta}{9\delta} - \frac{\frac{9}{8}\delta}{\rho}} < \frac{\eps}{\frac{4}{9} - \frac{\frac{9}{8}\cdot \frac{2\rho}{27}}{\rho}} = \frac{36\eps}{13} = \mys; 
\end{align*} 
and  
\begin{align*}
\frac{3\eps}{1 + \frac{r}{\rho}} \ge \frac{3\eps}{1 + \frac{\delta+\eps}{\rho}} \ge \frac{3\eps}{1+\frac{\frac{9}{8}\delta}{\rho}} > \frac{3\eps}{1 + \frac{\frac{9}{8}\cdot \frac{2\rho}{27}}{\rho}} = \frac{36\eps}{12} = \mys. 
\end{align*}

\end{document}